\documentclass[10pt, journal,twocolumn]{IEEEtran}
\usepackage{graphicx}
\usepackage{amsmath}
\usepackage{amssymb}
\usepackage{authblk}
\usepackage{cite}
\usepackage{blindtext}
\usepackage{subcaption}
\usepackage{float} 
\usepackage[utf8]{inputenc}
\newtheorem{remark}{Remark}

\newtheorem{theorem}{Theorem}[section]
\newtheorem{lemma}[theorem]{Lemma}
\newtheorem{proof}{Proof}
\newtheorem{assumption}{Assumption}

\usepackage{xcolor}
\usepackage{relsize}
\title{ \LARGE  
Event-Triggered Nonlinear Model Predictive Control for Cooperative Cable-Suspended Payload Transportation with Multi-Quadrotors}
\author{
Tohid Kargar Tasooji, IEEE Member, Sakineh Khodadadi, Guangjun Liu, Senior Member, IEEE
\thanks{ This work was supported by the Natural Sciences and Engineering Research
Council (NSERC) of Canada.
Tohid Kargar Tasooji is with the Department of Aerospace Engineering,
Toronto Metropolitan University, Toronto, ON M5B 2K3, Canada (e-mail:
tohid.kargartasooji@torontomu.ca). Sakineh Khodadadi is with the Department of Electrical and Computer
Engineering, University of Alberta, Edmonton,AB,  AB T6G 1H9, Canada (email: sakineh@ualberta.ca).
Guangjun Liu is with the Department of Aerospace Engineering,
Toronto Metropolitan University, Toronto, ON M5B 2K3, Canada (e-mail:
gjliu@torontomu.ca).
Guanghui Wang is with the Department of Computer Science,
Toronto Metropolitan University, Toronto, ON M5B 2K3, Canada (e-mail:
wangcs@torontomu.ca).}
}

\begin{document}
\maketitle 
\begin{abstract}
Autonomous Micro Aerial Vehicles (MAVs), particularly quadrotors, have shown significant potential in assisting humans with tasks such as construction and package delivery. These applications benefit greatly from the use of cables for manipulation mechanisms due to their lightweight, low-cost, and simple design. However, designing effective control and planning strategies for cable-suspended systems presents several challenges, including indirect load actuation, nonlinear configuration space, and highly coupled system dynamics. In this paper, we introduce a novel event-triggered distributed Nonlinear Model Predictive Control (NMPC) method specifically designed for cooperative transportation involving multiple quadrotors manipulating a cable-suspended payload. This approach addresses key challenges such as payload manipulation, inter-robot separation, obstacle avoidance, and trajectory tracking, all while optimizing the use of computational and communication resources. By integrating an event-triggered mechanism, our NMPC method reduces unnecessary computations and communication, enhancing energy efficiency and extending the operational range of MAVs. The proposed method employs a lightweight state vector parametrization that focuses on payload states in all six degrees of freedom, enabling efficient planning of trajectories on the SE(3) manifold. This not only reduces planning complexity but also ensures real-time computational feasibility. Our approach is validated through extensive simulation, demonstrating its efficacy in dynamic and resource-constrained environments.
\end{abstract}

\section{INTRODUCTION}
\IEEEPARstart{M}ICRO Aerial Vehicles (MAVs) equipped with onboard sensors hold significant promise for assisting or taking over human roles in intricate or hazardous missions like exploration [1], inspection [2], mapping [3], environmental interaction [4], search and rescue [5], and transportation [6], [7]. Aerial transportation, in particular, presents a swifter and more adaptable alternative to ground transportation, especially in congested urban settings. Moreover, a fleet of aerial robots can deliver supplies and establish communication in regions where GPS signals are sporadic or absent. To enhance payload capacity, one can opt for either larger aerial vehicles or employ a team of MAVs working together to transport cargo. Although the system's complexity escalates with the number of robots involved, a MAV team potentially offers greater resilience to the mission compared to a solitary MAV, particularly in instances of vehicle malfunctions.

In recent years, cooperative transportation utilizing quadrotors for cable-suspended payloads has emerged as a promising field of research. Several notable reported works have contributed significantly to this area, each offering unique insights and methodologies. 
 Loianno et al. [8] investigate cooperative transportation using small quadrotors that are equipped with a single camera and an inertial measurement unit (IMU). They tackle challenges related to maintaining system stability and accurately estimating the position and orientation (pose) of the quadrotors.
 It introduces novel approaches to coordinated control and cooperative localization, validated through experimental results. Unlike previous methods relying on linearized controllers or external systems, this work offers simpler solutions using nonlinear controllers and rigid connections. These advancements address technical challenges like complex dynamics and real-time estimation, enabling autonomous transportation in diverse environments.  Li et al. [9] address the technical challenges of cooperative transportation of cable-suspended rigid body payloads using MAVs. It proposes a distributed vision-based coordinated control system, enabling independent control of each MAV and distributed estimation of cable direction and velocity. Additionally, a cooperative estimation scheme is introduced, allowing inference of the payload’s full six Degrees of Freedom (DoF) states by sharing local position estimates and relative positions among the team of MAVs. Scalability considerations are also discussed, indicating feasibility in real-world scenarios such as warehouses and GPS-denied environments. 
 
 Li et al. [10], 
 present a novel Nonlinear Model Predictive Control (NMPC) method for quadrotors to manipulate rigid-body payloads via suspended cables in six DoF. It tackles challenges like indirect load actuation and highly coupled system dynamics, utilizing system redundancies for tasks like obstacle avoidance. Employing a lightweight state vector and planning trajectories on the SE(3) manifold, the method ensures real-time computation and scalability. Jin et al. [11] propose a constrained cooperative control framework for multiple Unmanned Aerial Vehicles (UAVs) transporting a three-dimensional load via cables. This framework specifically addresses performance constraints related to payload position tracking error and safety constraints involving the relative distance between quadrotors and obstacle avoidance. It effectively handles various challenges using universal barrier functions to manage inter-UAV distances and payload movement. Adaptive estimators manage uncertainties in UAV inertia matrices, ensuring compliance with constraint requirements. The framework guarantees exponential convergence in tracking errors while demonstrating efficacy through simulation. 
 
 Li et al. [12] present a new simulator for aerial transportation and manipulation with quadrotor MAVs using passive mechanisms. It includes models, planning, and control algorithms, alongside a groundbreaking collision model for interactions between quadrotors and payloads via cables. Through thorough simulation and real-world experiments, the authors validate the effectiveness of their approach, advancing our understanding and application of MAV technology.
 Erunsal et al. [14] compare linear and nonlinear Model Predictive Control strategies for MAVs trajectory tracking. It establishes fair testing conditions and uses identical algorithms, and parameters for both strategies. To address uncertainties, it employs parameter identification and a disturbance observer. Through extensive experiments, it proposes a decision-making framework for MPC selection based on trajectory characteristics and available resources, filling a gap in the literature.

In this work, our primary focus is on multi-quadrotor systems characterized by constrained energy and computational resources. Particularly, within the realm of multi-quadrotor systems, especially those with limited energy resources such as MAVs, optimizing energy utilization is of paramount importance. Typically, each robot in such systems is equipped with a compact embedded microprocessor responsible for managing sensor sampling, inter-robot communication, and controller updates. Traditionally, these systems rely on a classical time-triggered approach, where robots exchange state information periodically to update their control signals. However, this periodic exchange can lead to unnecessary communication and energy expenditure, which is particularly critical for MAVs with constrained battery capacity. Additionally, constraints like limited communication bandwidth can exacerbate issues such as packet dropouts and delays. To mitigate these challenges, there is a growing interest in event-triggered control. Unlike the time-triggered method, event-triggered control updates are triggered by specific events rather than occurring at fixed intervals [13], [18]-[22], [28]. This approach ensures system stability and performance while minimizing unnecessary communication and energy consumption, thereby enhancing the efficiency of MAVs and similar energy-constrained systems. Event-triggered control introduces a mechanism where control updates are only triggered when certain conditions are met, such as significant deviations from desired trajectories or changes in system dynamics. By implementing event-triggered control, we can reduce the computational load on the onboard processors of quadrotors, conserve energy, and extend the operational lifetime of the system. Moreover, event-triggered control can improve the overall system response by allocating computational resources more efficiently, thereby enhancing the robustness and scalability of cooperative transportation tasks.

Here are some challenges and gaps regarding the cooperative transportation of cable-suspended payloads with multi-quadrotors:
\begin{itemize}
    \item The first challenge lies in developing a control algorithm that facilitates the scalability of cooperative transportation systems for cable-suspended payloads in real-world scenarios. Meeting this challenge requires deployment strategies that seamlessly integrate multi-quadrotor systems into existing infrastructure, all while managing computational resources, communication bandwidth, and various operational constraints.

    \item The second challenge is how to design a control algorithm that can adapt to changes in payload weight, shape, and size during transportation while ensuring the stability and maneuverability of the quadrotors despite variations in payload characteristics.
    
    \item The third challenge is how to achieve precise synchronization among multiple quadrotors to minimize cable oscillations and maintain payload stability.

    \item The final challenge is how to manage various constraints, such as actuator limits, obstacle avoidance, and inter-robot constraints while ensuring that the payload position tracking error satisfies performance criteria.
\end{itemize}

This article presents a novel event-triggered distributed NMPC approach tailored to address cooperative transportation challenges with multiple quadrotors and cable-suspended payloads. This strategy effectively manages tasks like payload manipulation, inter-robot separation, obstacle avoidance, and trajectory tracking while ensuring efficient operation and respecting actuator constraints. Furthermore, this approach optimizes scalability, computational and communication resources, and energy consumption for MAVs involved. By minimizing resource overhead and energy consumption through adaptive task allocation and reduced communication exchanges, the system extends MAVs' operational range and endurance, making it well-suited for dynamic and resource-constrained environments. Overall, it offers a robust solution for complex cooperative transportation tasks while enhancing efficiency and sustainability. The main contributions of this work
can be summarized as follows.
\\
\begin{itemize}
    \item Our work integrates an event-triggered mechanism into the NMPC framework, reducing optimization computations frequency while maintaining control performance. By dynamically allocating resources, it enhances scalability and efficiency for real-time applications with limited processing capabilities.

    \item Our solution adapts dynamically to payload changes, triggering optimization updates selectively based on payload characteristics or environmental conditions. This ensures responsive and effective control actions for managing diverse payloads, crucial for real-world scenarios with varying dynamics.
    \item Our solution handles actuator limits, obstacle avoidance, and inter-robot constraints while ensuring the payload position tracking error meets performance criteria.

    \item Finally, our methodology enables scalable and practical deployment in real-world scenarios, leveraging event-triggered mechanisms for computational efficiency, dynamic adaptation, and robustness. It facilitates efficient and reliable distributed payload transportation in challenging environments like warehouses and GPS-denied areas, effectively addressing scalability and deployment challenges.
\end{itemize}
The rest of this article is organized as follows. In Section II, we formulate the problem to be solved, including the problem of interest, system dynamics, and system constraints and requirements. In Section III, we provide our event-based NMPC methodology. In Section IV, we conduct simulations to demonstrate the effectiveness of the proposed method. Finally, in Section V, we present our conclusions.
\begin{figure}[h!]
\captionsetup{justification=centering}
 \centering \includegraphics[width=0.5 \textwidth]{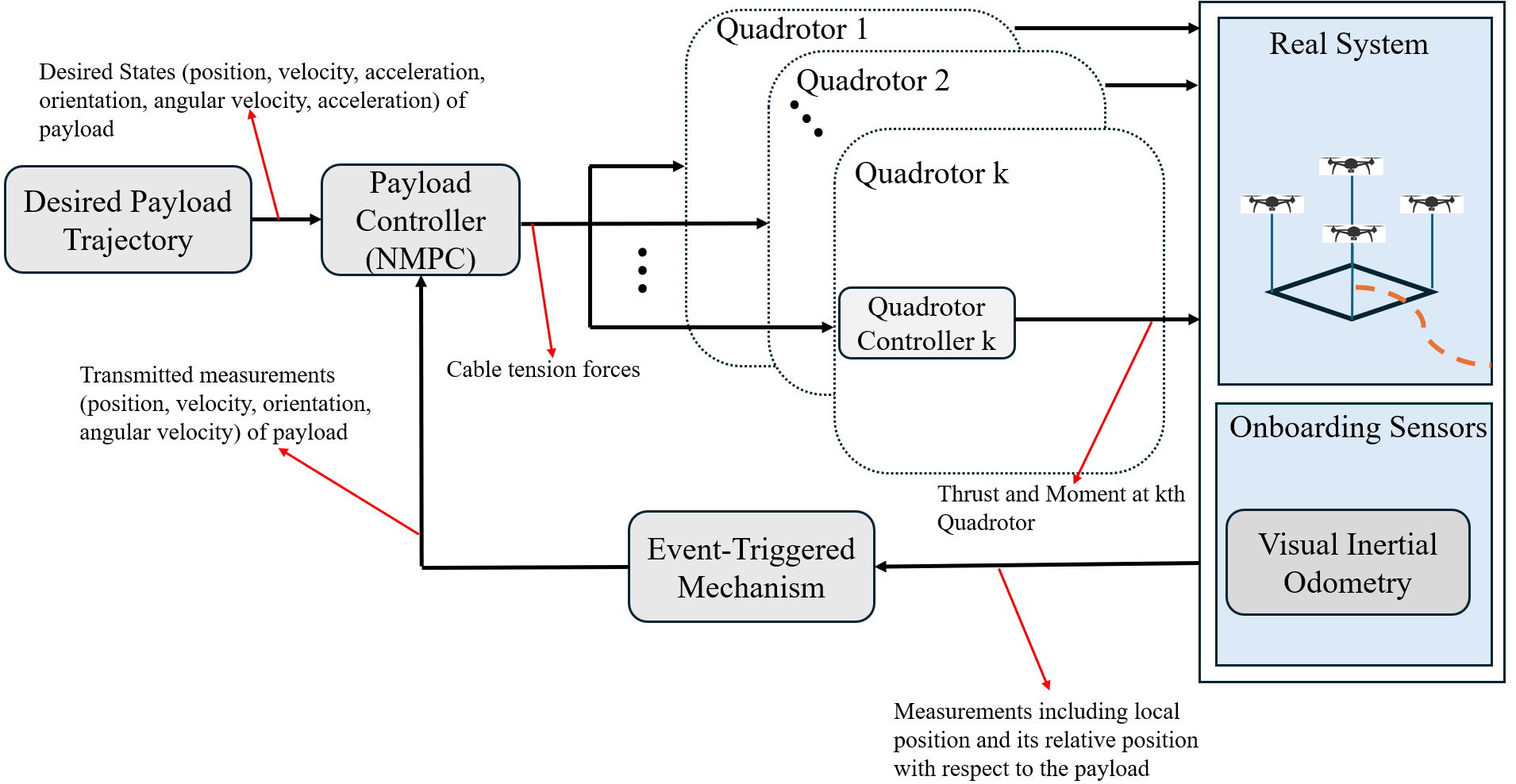}
  \caption{The control block diagram of the proposed approach}  \label{Fig1}
\end{figure} 
\section{PROBLEM FORMULATION}
\subsection{Problem of Interest}
 In the realm of cooperative transportation tasks using aerial robots, the need for efficient and scalable control strategies is paramount. This is particularly crucial for scenarios involving multiple quadrotors tasked with maneuvering a rigid-body payload suspended by cables, where ensuring accurate trajectory tracking of the payload's center of mass, maintaining inter-robot separation, avoiding obstacles, and adhering to the complex dynamics of the system in all 6 DoF are primary challenges. Achieving these objectives requires managing the constraints imposed by the quadrotor actuators and optimizing the distribution of cable tension forces. Traditional control methods, which often rely on continuous feedback mechanisms and frequent optimization problem-solving, struggle to scale efficiently with increasing system complexity and the number of agents involved.

To address these challenges, we introduce a groundbreaking event-triggered NMPC technique tailored for a fleet of quadrotors. Our NMPC framework is designed not only to handle the primary objective of payload manipulation but also to integrate additional functionalities such as inter-robot separation, obstacle avoidance, and desired trajectory tracking. This is achieved while ensuring compliance with the intricate dynamics governing the payload in all 6 DoF. A significant aspect of our approach is the optimization of cable tension forces, where a quadrotor controller utilizes the initial tension vector extracted from the projected horizon trajectory of the NMPC. This tension vector is then used to generate the necessary thrust and moment commands to execute the corresponding tension forces.
The proposed event-triggered NMPC method updates control actions only when necessary, reducing computational demands and enhancing scalability without compromising performance. This makes it suitable for managing multiple quadrotors and various payloads, providing robustness against uncertainties and disturbances. Designed for practical deployment in environments like warehouses and GPS-denied areas, it ensures efficient and reliable operation in dynamic settings.

\subsection{System Dynamics}
Consider a team of N (N $\geq$ 3) MAVs cooperatively transporting a rigid body payload connected via cables (refer to Figure 2), where the dynamics are described as [11]:
\begin{equation}
\text{MAVs}
\left\{
                \begin{array} {l} 
                m_{i}\ddot{p}_{i}(t) = - \text{sat}(F_{i}(t))R(\Theta_{i}(t))e_{z} + m_{i}ge_{z} \\ \ \ \ \ \ \ \ \ \ \ \ \  \ \ \   + T_{i}(t)R(\Theta_{L}(t))e_{i}(t) \\ 
                \dot{\Theta}_{i}(t) = \Gamma(\Theta_{i}(t))\omega_{i}(t) \\ 
                J_{i}\dot{\omega}_{i}(t) + \mathbb{S}(\omega_{i}(t)) J_{i} \omega_{i}(t) = \tau_{i}(t),
                \end{array}
\right.
\end{equation}

\begin{equation}
\text{Load}
\left\{
                \begin{array} {l} 
                m_{L}\ddot{p}_{L}(t) = m_{L}g e_{z} - \mathlarger{\mathlarger{\mathlarger{\mathlarger{\sum}}}}_{i=1}^{N} T_{i}(t)R(\Theta_{L}(t))e_{i}(t) \\ 
                \dot{\Theta}_{L}(t) = \Gamma(\Theta_{L}(t))\omega_{L}(t)\\
                J_{L}\dot{\omega}_{L}(t) + \mathbb{S}(\omega_{L}(t)) J_{L} \omega_{L}(t) = \mathlarger{\mathlarger{\mathlarger{\mathlarger{\sum}}}}_{i=1}^{N} \mathbb{S}(r_{i})(-T_{i}(t)e_{i}(t))
                \end{array}
\right.
\end{equation}
where $m_i \in \mathbb{R}^+$ is the mass of the $i$th quadrotor $(i = 1, \ldots, N)$, and $J_i \in \mathbb{R}^{3 \times 3}$ is a symmetric positive definite matrix representing the inertia. The position and attitude in the inertial reference frame are represented as $p_i(t) = [x_i(t), y_i(t), z_i(t)]^T \in \mathbb{R}^3$ and $\Theta_i(t) = [\phi_i(t), \theta_i(t), \psi_i(t)]^T \in \mathbb{R}^3$, respectively. $R(\Theta_i(t)) \in SO(3)$ is the rotation matrix, which relates the body-fixed frame to the inertial frame and is expressed as
\begin{equation}
R(\Theta_i) =
\begin{bmatrix}
c\phi_i c\theta_i & c\phi_i s\theta_i s\psi_i - c\psi_i s\phi_i & c\psi_i c\phi_i s\theta_i + s\phi_i s\psi_i \\
c\theta_i s\phi_i & c\phi_i c\psi_i + s\phi_i s\theta_i s\psi_i & c\psi_i s\phi_i s\theta_i - c\phi_i s\psi_i \\
-s\theta_i & c\theta_i s\psi_i & c\theta_i c\psi_i
\end{bmatrix}
\end{equation}
The angular velocities relative to this body-fixed frame are represented by $\omega_i(t) = [\omega_{xi}(t), \omega_{yi}(t), \omega_{zi}(t)]^T \in \mathbb{R}^3$, and $\Gamma(\Theta_{i}(t))$ is the transformation matrix that links the angular velocity in the body-fixed frame to the derivative of the Euler angles in the inertial frame and is expressed as
\begin{equation}
\Gamma(\Theta_i) =
\begin{bmatrix}
1 & s\phi_i t\theta_i & c\phi_i t\theta_i \\
0 & c\phi_i & -s\phi_i \\
0 & \frac{s\phi_i}{c\theta_i} & \frac{c\phi_i}{c\theta_i}
\end{bmatrix}
\end{equation}
which is well-defined and invertible when $-\pi/2 < \phi_i(t) < \pi/2$ and $-\pi/2 < \theta_i(t) < \pi/2$. Furthermore, $g \in \mathbb{R}$ is the gravitational acceleration and $\mathbf{e}_z = [0, 0, 1]^T \in \mathbb{R}^3$ is the unit vector. Next, $T_i(t) \in \mathbb{R}^+$ represents the tension in the $i$th rigid cable, $\text{sat}(a)$ denotes the saturation function, where $a \in \mathbb{R}^+$ and $F_i(t) \in \mathbb{R}^+$ represents the thrust of the $i$th quadrotor $F_i(t) \in \mathbb{R}^+$ $(i = 1, \ldots, N)$ which is subject to saturation nonlinearity described in
\begin{equation}
\text{sat} (F_i(t)) =
\begin{cases}
F_{\text{max}}, & \text{if } F_i(t) \geq F_{\text{max}} \\
F_i(t), & \text{otherwise}
\end{cases}
\end{equation}
where $F_{\text{max}_i}$ is the saturation limit for thrust $F_i(t)$ and $\text{sign}(\cdot)$ is the sign function. Finally, $\tau_i(t) \in \mathbb{R}^3$ represents the torques of the $i$th quadrotor $(i = 1, \ldots, N)$.\\
Similarly, $m_L \in \mathbb{R}^+$ is the load mass, and $J_L \in \mathbb{R}^{3 \times 3}$ is the load inertia that is symmetric and positive definite, where the subscript $L$ stands for “Load” $\mathbf{p}_L(t) = [x_L(t), y_L(t), z_L(t)]^T \in \mathbb{R}^3$ and $\mathbf{\Theta}_L(t) = [\phi_L(t), \theta_L(t), \psi_L(t)]^T \in \mathbb{R}^3$ represent the load position and attitude in the inertial reference frame, respectively, and $\boldsymbol{\omega}_L(t) = [\omega_{xL}(t), \omega_{yL}(t), \omega_{zL}(t)]^T \in \mathbb{R}^3$ represents the load rotational velocity with respect to its body-fixed frame. Furthermore, as shown in Figure 2, $\mathbf{r}_i \in \mathbb{R}^3$ is the attachment point on the payload by the $i$th link, represented in the payload body-fixed frame. Finally, $\mathbf{e}_i(t) \in S^2$ is the unit direction vector from the $i$th MAV mass center towards the $i$th link attachment point. 
\subsection{System Constraint Requirements}
Performance Constraints Requirement: In the cooperative transportation task, the payload is expected to follow a desired trajectory, represented by $\mathbf{p}_d^L(t) = [x_d^L(t), y_d^L(t), z_d^L(t)]^T \in \mathbb{R}^3$. Additionally, all MAVs are required to follow a desired formation pattern, where the coordinates of the reference trajectory for the $i$th vehicle where $i = 1, \ldots, N$, denoted by $\mathbf{p}_d^i(t) = [x_{d}^{i}(t), y_{d}^{i}(t), z_{d}^{i}(t)]^T \in \mathbb{R}^3$. Now, define the line-of-sight (LOS) distance tracking error for the payload as
\begin{equation}
e^L(t) = \sqrt{(x_L - x_{d}^L)^2 + (y_L - y_{d}^L)^2 + (z_L - z_{d}^L)^2}
\end{equation}
which measures the distance between the desired and actual position of the payload. Furthermore, for the $i$th quadrotor ($i = 1, \ldots, N$), define the line-of-sight distance tracking error $e^i(t)$ as
\begin{equation}
e^i(t) = \sqrt{(x_i - x_{d}^i)^2 + (y_i - y_{d}^i)^2 + (z_i - z_{d}^i)^2}.
\end{equation}

 During cooperative transportation, there are certain system constraint requirements that need to be satisfied. In order to ensure the precise and safe functioning of the system, several constraints need to be satisfied. Firstly, the payload position tracking error $e^L(t)$ must adhere to the performance constraint requirement:
\begin{equation}
e^L(t) < \epsilon_d^H(t), \quad 
\end{equation}
where $\epsilon_d^H(t) > 0$ is a time-varying constraint function, which is user-defined and belongs to the class $C^3$. This constraint ensures that the payload does not deviate significantly from its desired trajectory.

Similarly, the tracking error $e^i(t)$ for each MAV must meet the user-defined performance constraint:
\begin{equation}
e^i(t) < \epsilon_{d}^{H_{i}}(t),
\end{equation}
where $\epsilon_{d}^{H_{i}}(t) > 0$ is a time-varying constraint function, also belonging to the class $C^3$. This constraint ensures that each MAV stays close to its desired trajectory.

Inter-robot Separation Constraints: We first define the desired LOS relative distance between any two quadrotors $\delta_{ij}^d(t)$ as follows
\begin{equation}
\delta_{ij}^d(t) = \sqrt{(x_{d}^i - x_{d}^j)^2 + (y_{d}^i - y_{d}^j)^2 + (z_{d}^i - z_{d}^j)^2},
\end{equation}
and the actual LOS relative distance $\delta_{ij}(t)$ is given by
\begin{equation}
\delta_{ij}(t) = \sqrt{(x_i - x_j)^2 + (y_i - y_j)^2 + (z_i - z_j)^2}.
\end{equation}
Next, the line-of-sight relative distance tracking error between the $i$th and $j$th quadrotors (where $i, j = 1, \ldots, N$, and $j \neq i$) is defined as $e^{ij}(t) = d_{ij}(t) - \delta_{ij}(t)$. This error must satisfy the safety constraint:
\begin{equation}
-\epsilon_{W_{ij}}(t) < e^{ij}(t) < \epsilon_{H_{ij}}(t),
\end{equation}
where $\epsilon_{H_{ij}}(t) > 0$ represents the upper constraint for $e^{ij}(t)$, and $-\epsilon_{W_{ij}}(t) < 0$ is the lower bound, with $0 < \epsilon_{W_{ij}}(t) < \delta_{ij}(t)$. Both $\epsilon_{H_{ij}}(t)$ and $\epsilon_{W_{ij}}(t)$ belong to the class $C^3$. This constraint ensures that the inter-quadrotor distance remains within acceptable bounds.

Constraints for Avoiding Obstacles: Given the position of an obstacle as $\mathbf{p}_O(t) = [x_O(t), y_O(t), z_O(t)]^T$, we define the line-of-sight (LOS) distance between quadrotors and payload as
\begin{equation}
e^{LO}(t) = \sqrt{(x_L - x_O)^2 + (y_L - y_O)^2 + (z_L - z_O)^2}.
\end{equation}
the following constraints can be formulated to ensure that both the robot team and the payload maintain a secure distance from obstacles in the global frame:
\begin{equation}
e^{LO}(t) \geq \epsilon_{OL}
\end{equation}
Actuator Constraints:
As we obtain the predicted cable
tension force, we can further limit the tension
force norm to provide some boundary for the actuators
\begin{equation}
\lVert T_{i} \rVert \leq f_{max}
\end{equation}
\begin{remark}
Both (8) and (9) are considered performance constraint requirements. They indicate that during cooperative load transportation, both the load and the MAV team should closely follow their desired trajectories. Violating these performance constraint requirements can lead to failure in maintaining the desired formation and/or collisions with environmental boundaries. Constraint (12) belongs to the safety constraint. The physical meaning of this requirement is that any two MAVs in the team cannot be too close or too far apart, which would lead to inter-MAV collisions or overstretching of the suspension cables, respectively. Constraint (14) belongs to the safety distance between the obstacle and payload. Also, constraint (15) represents the boundary on tension force.    
\end{remark}

\section{Methodology}
In this section, we present the hierarchical controller design tailored for systems comprising three or more quadrotors manipulating a rigid-body payload via cables. 

\subsection{Event-triggered Mechanism}
The event-triggering condition in the context of cooperative transportation of cable-suspended payloads with multi-quadrotors is crucial for determining when to initiate control updates. This condition is based on comparing the current state measurement with its optimal prediction from the previous triggering event. When the error between them surpasses a predefined threshold, it signals the need for recalculating the control sequence.

The event-triggering condition can be defined as follows:

Let \( \xi_f(k) \) represent the current state measurement at time step \( k \), and \( \xi^*_f(k|k_j) \) denote its optimal prediction obtained at the previous triggering time \( k_j \). The event is triggered when the following condition is violated:

\begin{equation}
\| \xi_f(k) - \xi^*_f(k|k_j) \| > \alpha \| \xi_f(k) \| + \beta
\end{equation}

where \( \alpha \) and  \( \beta \) are the event-triggered parameters and \( \| \cdot \| \) denotes the norm of the error. This condition ensures that a new control sequence is calculated only when the deviation between the current state and its optimal prediction exceeds a certain acceptable limit.

In addition to the event-triggering condition, the prediction horizon update strategy plays a vital role in adapting the NMPC scheme to the system dynamics. This strategy ensures that the prediction horizon is adjusted dynamically based on the interexecution time and the prediction made at the previous triggering event. By limiting the shrinking size of the horizon, closed-loop stability of the system is maintained while ensuring the feasibility of the Optimal Control Problem (OCP) at each triggering time.

The prediction horizon update strategy can be described as follows:

The length of the prediction horizon, \( N_k \), is determined based on the interexecution time \( m_k \) and the prediction made at the previous triggering event \( k_j \). This relationship is expressed as:

\begin{equation}
N_k = g(m_k, k_j)
\end{equation}

where \( g(\cdot) \) is a function that determines the appropriate length of the prediction horizon based on the given inputs. 

Moreover, the development of the event-triggering condition and the prediction horizon update strategy are co-designed to ensure effective control. Specifically, the prediction horizon \( N_k \) at time \( k \) is related to the interexecution time \( m_k \) and the prediction horizon \( N_{k_j} \) at the previous triggering event \( k_j \). The relationship can be expressed as:

\begin{equation}
N_k = h(N_{k_j}, m_k)
\end{equation}

where \( h(\cdot) \) is a function that determines the updated prediction horizon based on the previous prediction horizon and the inter-execution time.

\subsection{Payload Nonlinear Model Predictive Control}
In this section, we introduce a novel NMPC approach designed for manipulating the pose of the payload. NMPC, in general, is a predictive control method that computes a sequence of system states $\{X_0, X_1, \ldots, X_N\}$ and inputs $\{U_0, U_1, \ldots, U_{N-1}\}$ over a predetermined time horizon of $N$ steps. This computation aims to optimize an objective function while adhering to nonlinear constraints and system dynamics. The objective function comprises a running cost, $h(X, U)$, and a terminal cost, $h_N(X)$, and is formulated as follows [23]-[27]:
\begin{equation}
\min_{X_0,\dots,X_N, U_0,\dots,U_{N-1}} \mathlarger{\mathlarger{\sum}}_{i=0}^{N-1} h(X_i,U_i) + h_N(X_N),\\
\end{equation}
subject to the constraints:
\begin{equation*}
X_{i+1} = f(X_i,U_i), \quad \forall i = 0, \dots, N - 1 \\ 
\end{equation*}
\begin{equation*}
X_0 = X(t_0),\\
\end{equation*}
\begin{equation*}
g(X_i,U_i) \leq 0,
\end{equation*}
where $X_{i+1} = f(X_i,U_i)$ represents the nonlinear system dynamics in discrete form, and $g(X,U)$ represents additional state and input constraints. The optimization occurs with the initial condition $X_0$ while respecting the system dynamics $f(X,U)$.

In the following, we present our proposed NMPC formulation for transporting a rigid-body payload with $n$ quadrotors using suspended cables. We discuss the advantages of the chosen cost function and describe the nonlinear system dynamics and constraints for exploiting secondary tasks (obstacle
avoidance and inter-robot spatial separation) and respecting actuator limits.

Cost Function:
First, let us define the state vector and input vector of the NMPC as 
\begin{equation}
X = \begin{pmatrix} p_L \\
\Theta_L \\
v_L \\
\omega_L \\ \end{pmatrix} , \quad U = \begin{pmatrix} F \\ M \end{pmatrix}
\end{equation}
where based on Eq. (2) \(F\)  is the sum of all the cable tension forces and \(M\) is the total moments generated by the cable forces on the payload. 
Since the task is to manipulate the payload to track a desired trajectory in \( SE(3) \), we choose the following objective function:
\begin{equation}
\min_{X_i,U_i} e^T _{X_N} Q _{X_N} e_{X_N} + \mathlarger{\mathlarger{\sum}}_{i=0}^{N-1} e^T _{X_i} Q _{X_i} e_{X_i} + e^T _{U_i} Q_{U} e_{U_i}
\end{equation}
where \(e_{X_i},  e_{U_i} \) are the state and wrench errors determined by comparing predicted states and inputs to trajectory references defined as:
\begin{equation}
e_{X_i} = \begin{pmatrix} p_{L,des} - p_L \\ v_{L,des} - v_L \\ \log(\Theta_L \otimes \Theta_{des,L}^{-1}) \\ \omega_{L,des} - \omega_L \end{pmatrix}_t^i , \quad e_{U_i} = \begin{pmatrix} F_{des} - F \\ M_{des} - M  \end{pmatrix}_t^i
\end{equation}
The proposed method introduces a streamlined approach to Nonlinear Model Predictive Control (NMPC), offering significant advantages over existing techniques. By simplifying the planning process to focus solely on payload states and their derivatives in SE(3), the complexity associated with trajectory planning for cable directions and tension is eliminated. This results in a four-fold reduction in optimization dimensionality compared to previous methods, making the approach more computationally efficient. Additionally, the inclusion of a null space coefficient vector allows for the integration of secondary tasks such as obstacle avoidance, enhancing the system's versatility. The proposed NMPC is solved using Sequential Quadratic Programming (SQP) in real-time, with HPIPM selected as the solver for quadratic programming within SQP.

System Dynamics: Based on equations (1) and (2), we derive the dynamic equation for the state \( X \) as 
\begin{equation}
\dot{X} = f(X,U) = \begin{pmatrix} \dot{p}_L \\ \Theta_{L} \otimes \omega_L \\ \frac{1}{m_L}F - g \\ J_L^{-1}(M - \omega_L \times J_L\omega_L) \end{pmatrix}
\end{equation}
To integrate the dynamic equation into the discrete-time formulation, we employ the 4th order Runge-Kutta method to numerically integrate \( \dot{X} \) over the sampling time \( dt \) given the state \( X_i \) and input \( U_i \) as \( X_{i+1} = f_{RK4}(X_i,U_i, dt) \).
\subsection{Payload Event-triggered Nonlinear Model Predictive Control }
The condition for triggering events is formulated based on comparing the current state with its optimal prediction from the last triggering event. When the discrepancy between them exceeds a preset threshold, it indicates the need for recalculating a new control sequence. Initially, we analyze the error between the actual and optimal trajectories using a multi-step $m$ open-loop control from the same starting point:
\begin{equation}
\begin{array}{l}
\lVert \xi_f(k_j + m) - \xi_f^* (k_j + m|k_j) \rVert\\
= \bigg \lVert \xi_f(k_j) + \mathlarger{\mathlarger{\sum}}_{i=0}^{m-1} \left[\delta f_h\left(\xi_f(k_j + i), u^*_f(k_j + i)\right) + d(k_j + i)\right] \\ \ \ - \xi^*_f(k_j) - \mathlarger{\mathlarger{\sum}}_{i=0}^{m-1} \delta f_h\left(\xi^*_f(k_j + i), u^*_f(k_j + i)\right) \bigg \rVert
\end{array}
\end{equation}

Using the condition $\xi_f (k_j) = \xi^*_f(k_j)$ and Lipschitz condition yields
\begin{equation}
\begin{array}{l}
\lVert \xi_f(k_j + m) - \xi^*_f \left(k_j + m \middle| kj\right) \rVert \\
= L_P \mathlarger{\mathlarger{\sum}}_{i=0}^{m-1} \Big \lVert \left(\xi_f(k_j + i) - \xi^*_f(k_j + i)\right) \Big \rVert + m\eta.
\end{array}
\end{equation}
Applying the Gronwall–Bellman–Ou–Iang-type inequality given by Lemma [16] corresponds to
\begin{equation}
\begin{array}{l}
\Big \lVert \xi_f(k_j + m) - \xi^*_f(k_j + m) \Big \rVert \leq m\eta e^{L_P \delta (\sigma - 1)}
\end{array}
\end{equation}
This analysis leads to setting the threshold for the triggering condition $\sigma \eta e^{L_P \delta (\sigma - 1)}$ based on a tuning parameter, where $\sigma$ represents the minimum time between executions. The triggering condition is established as the difference between the current and predicted states exceeding the computed value:
\begin{equation}
\begin{array}{l}
\Big \lVert \xi_f(k_j + m_{k_{j}}) - \xi^*_f(k_j + m_{k_{j}}) \Big \rVert \geq \sigma \eta e^{L_P \delta (\sigma - 1)}
\end{array}
\end{equation}
The interval between executions is determined as
\begin{equation}
\begin{array}{l}
m_{k_{j}} = \textbf{sup}_m \bigg \{ \Big \lVert \xi_f(k_j + m) - \xi^*_f(k_j + m) \Big \rVert \geq \\ \ \ \ \  \ \ \ \ \  \alpha \| \xi_f(k_j + m)) \| + \beta \bigg  \}
\end{array}
\end{equation}
the shortest duration where this difference remains below the threshold.

Additionally, the system is automatically triggered at the initial time $k_{0}$ and at $k_{j} + N_{k_{j}}$ to account for scenarios where no control action is available. As a result, the inter-execution time is upper and lower bounded by $\sigma \leq  m_{k_j} \leq  N_{k_j}$.
\begin{remark}
Increasing the minimum inter-execution time $\sigma$ has the potential to decrease the triggering frequency; however, this reduction often comes with a compromise in tracking performance. Thus, selecting an appropriate value for $\sigma$ entails finding a balance between these competing factors.
\end{remark}

A long prediction can ensure stability but may complicate the Optimal Control Problem (OCP). Reducing the prediction horizon may simplify the OCP by decreasing its dimensionality. In standard Model Predictive Control (MPC), the prediction horizon remains constant to ensure that the tracking error at its end enters the terminal region. However, as the tracking error approaches the terminal region, a shorter prediction horizon may suffice to satisfy the terminal constraint. This forms the basis of the shrinking strategy for the prediction horizon described below. At each triggering instant, the previous prediction horizon's length can inform the design of the shrinking strategy. The shortest horizon length from the previous prediction is determined by:
\begin{equation}
\begin{array}{l}
\hat{N}_{k_j} = \inf \{ i : p^*_e(kj + i|k_j) \in \Omega_{\epsilon}, i \in \mathbb{N}_{[0, N_{k_j} - 1]} \}
\end{array}
\end{equation}
However, excessively shortening the prediction horizon may result in the infeasibility of the OCP and, consequently, destabilize the system. Therefore, the stability condition of the closed-loop system must be considered. To ensure the stability of the closed-loop system, it is necessary that $k_{j+1} + N_{k_j+1} > k_j + N_{k_j}$. Consequently, an upper bound for the shrinking length is given by 
\begin{equation}
\begin{array}{l}
n_{k_j} \leq m_{kj} - 1
\end{array}
\end{equation}

To summarize, the prediction horizon is updated and decreased by 
\begin{equation}
\begin{array}{l}
n_{k_j} = \min \{ m_{k_j} - 1, N_{k_j} - \hat{N}_{kj} \}
\end{array}
\end{equation}
The development of the event-triggering condition and the prediction horizon update strategy are co-designed rather than simply combined. From Equation (31), one can observe that the reduced length of the horizon $n_{kj}$ is influenced by both the inter-execution time $m_{kj}$ and the prediction at the previous triggering instant. Equation (29) introduces a shrunken horizon $\hat{N}_{kj}$ that ensures the feasibility of the OCP at the current triggering time, while Equation (30) constrains the shrinking size of the horizon to establish closed-loop stability of the system. The prediction horizon at $k_j+1$ should satisfy 
\begin{equation}
\begin{array}{l}
k_j + N_{k_j} < k_j+1 + N_{k_j+1} \leq k_j+1 + N_{k_j}
\end{array}
\end{equation}
Simultaneously, the tracking error is ensured to enter the terminal region at $k_j+1 + N_{k_j+1}$. The following theorem states the main results of event-triggered MPC with an adaptive prediction horizon scheme.
\begin{assumption}
There exists a robust terminal set $\Omega_\epsilon$ and a corresponding controller $\kappa(e_{X_i})$, such that, $\forall e_{X_i}(k) \in \Omega_\epsilon$
\begin{equation}
g(e_{X_i}(k + 1)) - g(e_{X_i}(k)) \leq -L(e_{X_i}(k), \kappa(\hat{x}(k))).
\end{equation}
\end{assumption}
\begin{lemma}
The nonlinear function $f_h(\xi, u)$ is locally Lipschitz continuous in $\xi$ with Lipschitz constant $L_P = \sqrt{2(a^2 + \rho^2 b^2)}$.    
\end{lemma}
\begin{theorem}
Suppose that the payload is controlled by the optimizing control \( u^*_f(k_j + i | k_j) \) at time \( k_j + i \), the controller update time is determined by equation (27), and the prediction horizon is decreased according to equation (31). Then, the payload system (23) is input-to-state stable.
\end{theorem}
\begin{proof}
Take the optimal cost as a Lyapunov function and consider the difference of the cost between $k_j$ and $k_{j+1}$:
\begin{equation}
\begin{array}{l}
V(k_{j+1}) - V(k_j) = J_d(e_{X}^*(k_{j+1}), e_{U}^*(k_{j+1}), N_{k_{j+1}})  \\ 
- J_d(e_{X}^*(k_j), e_{U}^*(k_j), N_{k_j})  \\
= \mathlarger{\mathlarger{\sum}}_{i=0}^{N_{k_{j+1}} - 1} \left( \lVert e_{X}^*(k_{j+1} + i | k_{j+1}) \rVert_Q^2 + \lVert e_{U}^*(k_{j+1} + i | k_{j+1}) \rVert_P^2 \right)  \\
+ \lVert e_{X}^*(k_{j+1} + N_{k_{j+1}} | k_{j+1}) \rVert_R^2 \\
- \mathlarger{\mathlarger{\sum}}_{i=0}^{N_{k_j} - 1} \left( \lVert e_{X}^*(k_j + i | k_j) \rVert_Q^2 + \lVert e_{U}^*(k_j + i | k_j) \rVert_P^2 \right)  \\
- \lVert e_{X}^*(k_j + N_{k_j} | k_j) \rVert_R^2.
\end{array}
\end{equation}
Decomposing the terms and using the event-triggering condition, we get:
\begin{equation}
\begin{array}{l}
\mathlarger{\mathlarger{\sum}}_{i=0}^{N_{k_{j+1}} - 1} \left( \lVert e_{X}^*(k_{j+1} + i | k_{j+1}) \rVert_Q^2 + \lVert e_{U}^*(k_{j+1} + i | k_{j+1}) \rVert_P^2 \right)  \\
= \mathlarger{\mathlarger{\sum}}_{i=0}^{m_{k_j} - 1} \left( \lVert e_{X}^*(k_j + i | k_j) \rVert_Q^2 + \lVert e_{U}^*(k_j + i | k_j) \rVert_P^2 \right)  \\
\quad + \mathlarger{\mathlarger{\sum}}_{i=m_{k_j}}^{N_{k_j} - 1} \left( \lVert e_{X}^*(k_j + i | k_{j+1}) \rVert_Q^2 - \lVert e_{X}^*(k_j + i | k_j) \rVert_Q^2 \right) \\
\quad + \mathlarger{\mathlarger{\sum}}_{i=N_{k_j}}^{N_{k_{j+1}} - 1} \left( \lVert e_{X}^*(k_j + i | k_{j+1}) \rVert_Q^2 + \lVert e_{U}^*(k_j + i | k_{j+1}) \rVert_P^2 \right) \\
\quad + \lVert e_{X}^*(k_{j+1} + N_{k_{j+1}} | k_{j+1}) \rVert_R^2 - \lVert e_{X}^*(k_j + N_{k_j} | k_j) \rVert_R^2.
\end{array}
\end{equation}
Then, the above can be rewritten as:
\begin{equation}
\begin{array}{l}
V(k_{j+1}) - V(k_j) = \\ 
- \mathlarger{\mathlarger{\sum}}_{i=0}^{m_{k_j} - 1} \left( \lVert e_{X}^*(k_j + i | k_j) \rVert_Q^2 + \lVert e_{U}^*(k_j + i | k_j) \rVert_P^2 \right) \\
\quad + \mathlarger{\mathlarger{\sum}}_{i=m_{k_j}}^{N_{k_j} - 1} \left( \lVert e_{X}^*(k_j + i | k_{j+1}) \rVert_Q^2 - \lVert e_{X}^*(k_j + i | k_j) \rVert_Q^2 \right)  \\
\quad + \mathlarger{\mathlarger{\sum}}_{i=N_{k_j}}^{N_{k_{j+1}} - 1} \left( \lVert e_{X}^*(k_j + i | k_{j+1}) \rVert_Q^2 + \lVert e_{U}^*(k_j + i | k_{j+1}) \rVert_P^2 \right)  \\
\quad + \lVert e_{X}^*(k_{j+1} + N_{k_{j+1}} | k_{j+1}) \rVert_R^2 - \lVert e_{X}^*(k_j + N_{k_j} | k_j) \rVert_R^2.
\end{array}
\end{equation}
Using the feasible control input construction technique, the triangle inequality, and the Gronwall–Bellman–Ou–Iang-type inequality, we have the following result:
\begin{equation}
\begin{array}{l}
V(k_{j+1}) - V(k_j) \leq  \\
- \mathlarger{\mathlarger{\sum}}_{i=0}^{m_{k_j} - 1} \left( \lVert e_{X}^*(k_j + i | k_j) \rVert_Q^2 + \lVert e_{U}^*(k_j + i | k_j) \rVert_P^2 \right) \\
\quad + \Delta(m_{k_j}, N_{k_j}, \eta) \\
\leq -\lVert e_{X}^*(k_j) \rVert_Q^2 + \Delta(\sigma, N_p, \eta),
\end{array}
\end{equation}
where
\begin{equation}
\begin{aligned}
\Delta(m_{k_j}, N_{k_j}, \eta) &= \mathlarger{\mathlarger{\sum}}_{i=m_{k_j}}^{N_{k_j} - 1} \left[ \sigma^2 \eta^2 \overline{q} e^{2L_P(i-1)} + 2m_{k_j} \eta \overline{q}^2 e^{L_P(i-1)} \right] \\
&\quad + \sigma \eta \overline{r} e^{L_P(N_{k_j}-1)}(r + \epsilon).
\end{aligned}
\end{equation}
$\Delta(\sigma, N_p, \eta)$ is clearly a $K_{\infty}$ function with respect to $\eta$, implying that the tracking error will converge to a neighborhood of the origin. The range of this neighborhood is related to the minimum inter-execution time $\sigma$ and the bound of the disturbance $\eta$.
\end{proof}

This framework is specifically designed for the cooperative transportation of cable-suspended payloads with multi-quadrotors. By incorporating the event-triggering condition and adaptive prediction horizon strategies, the system ensures both stability and performance while efficiently managing the computational load. The triggering mechanism ensures that control updates are performed only when necessary, optimizing resource usage. This is particularly beneficial in multi-quadrotor systems where communication and computational resources are limited. The adaptive prediction horizon allows the system to balance between the complexity of the OCP and the stability of the control process, adapting dynamically to the state of the system. The result is a robust, efficient, and scalable solution for cooperative payload transportation.
\begin{figure}[h!]
\captionsetup{justification=centering}
 \centering \includegraphics[width=0.25 \textwidth]{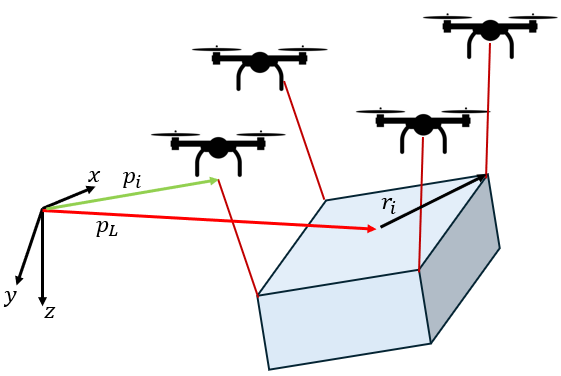}
  \caption{Cable-suspended load transportation by MAVs}  \label{Fig1}
\end{figure} 
\subsection{Quadrotor Control}
After solving the optimization problem outlined in equation (19), we obtain a series of system states $\{X_0, X_1, \ldots, X_N\}$ and inputs $\{U_0, U_1, \ldots, U_{N-1}\}$ over a fixed time horizon of $N$ steps. Selecting the initial input $U_0 = \begin{pmatrix} F \\ M \end{pmatrix}$  from the input sequence allows us to achieve the desired tension force for each cable of the quadrotor, denoted as $\mu_{\text{des}_k}$.
\begin{equation}
\begin{pmatrix}
\mu_{1,\text{des}} \\
\vdots \\
\mu_{n,\text{des}}
\end{pmatrix}
= \text{diag}(R_L, \ldots, R_L) P^\dagger \begin{pmatrix}
R^T_{L}F_0 \\
M_0
\end{pmatrix}
\end{equation}
Once these desired tension forces are determined, we select the tension input for each cable of the individual quadrotor by projecting the desired tension onto the corresponding cable, expressed as
\begin{equation}
\mu_k = \xi_k \xi_k^T \mu_{k,\text{des}}
\end{equation}
The desired direction $\xi_{k,\text{des}}$ and angular velocity $\omega_{k,\text{des}}$ for the $k$th cable link are then derived accordingly:
\begin{equation}
    \xi_{k,\text{des}} = -\frac{\mu_{k,\text{des}}}{\|\mu_{k,\text{des}}\|},  \ \ \ 
    \omega_{k,\text{des}} = \xi_{k,\text{des}} \times \dot{\xi}_{k,\text{des}},
\end{equation}
where $\dot{\xi}_{k,\text{des}}$ is the derivative of $\xi_{k,\text{des}}$.
The thrust $f_k$ and moments $M_k$ acting on the $k$th quadrotor are computed as follows:
\begin{equation}
\begin{array}{l} 
    f_k = u_k \cdot R_ke_3 = (u^{\parallel}_{k} + u^{\perp}_{k}) \cdot R_ke_3
    \end{array}
\end{equation}
\begin{equation}
\begin{array}{l} 
    M_k = K_Re_{R_k} + K_\Omega e_{\Omega_k} + \Omega_k \times J_k\Omega_k \\  \ \ \ \ \ \ \ \ \   - J_k(\hat{\Omega}_kR_k^T R_{k,\text{des}}\Omega_{k,\text{des}} - R_k^T R_{k,\text{des}}\dot{\Omega}_{k,\text{des}}).
\end{array}
\end{equation}
Here, $K_R$ and $K_\Omega$ are diagonal matrices with positive constants, and $e_{R_k}$ and $e_{\Omega_k}$ represent orientation and angular velocity errors, respectively.
\begin{equation}
\begin{array}{l} 
e_{R_k} = \frac{1}{2} \left( R_{k}^{\top} R_{k,\text{des}} - R_{k,\text{des}}^{\top} R_{k} \right)^\vee, \\
e_{\Omega_k} = R_{k}^{\top} R_{L, \text{des}} \Omega_{L,\text{des}} - \Omega_{L}.
\end{array}
\end{equation}

The inputs $u_{\perp k}$ and $u_{\parallel k}$ are designed according to specific formulations, incorporating terms such as cable direction, angular velocity, and acceleration.
\begin{equation}
\begin{array}{l} 
u^{\perp}_k = m_k l_k \xi_k - \Big [ K_{\xi_{k}} e_{\xi_k} - K_{\omega_k} e_{\omega_k} - \hat{\xi}_k^2 \omega_{k,\text{des}} \\ \ \ \ \ \ \ \ - (\xi_k \cdot \omega_{k,\text{des}}) \dot{\xi}_{k,\text{des}} \Big ] 
 - m_k \hat{\xi}_k^2 a_{k,c}, 
\end{array}
\end{equation}
\begin{equation}
\begin{array}{l} 
u^{\parallel}_k = \mu_k + m_k l_k \|\omega_k\|^2 \xi_k + m_k \xi_k \xi_k^\top a_{k,c}.
\end{array}
\end{equation}
where 
\begin{equation}
\begin{array}{l} 
a_{k,c} = \ddot{x}_{L,\text{des}} + g - R_L \hat{\rho}_k \dot{\Omega}_L + R_L \hat{\Omega}_L^2 \rho_k,
\end{array}
\end{equation}
where $K_{\xi_k}$ and $K_{\omega_k}$ are diagonal positive constant matrices,
$e_{\xi_k}$ and $e_{\omega_k}$ are the cable direction and cable angular
velocity errors respectively: 
\begin{equation}
\begin{array}{l} 
e_{\xi_k} = \xi_{k,\text{des}} \times \xi_k,
, \ \ \ e_{\omega_k} = \omega_k + \xi_k \times \xi_k \times \omega_{k,\text{des}}
\end{array}
\end{equation}
 For further details on stability analysis, readers are directed to reference [15].

\section{Simulations}
\begin{figure}[h]
\captionsetup{justification=centering}
 \centering \includegraphics[width=0.4 \textwidth]{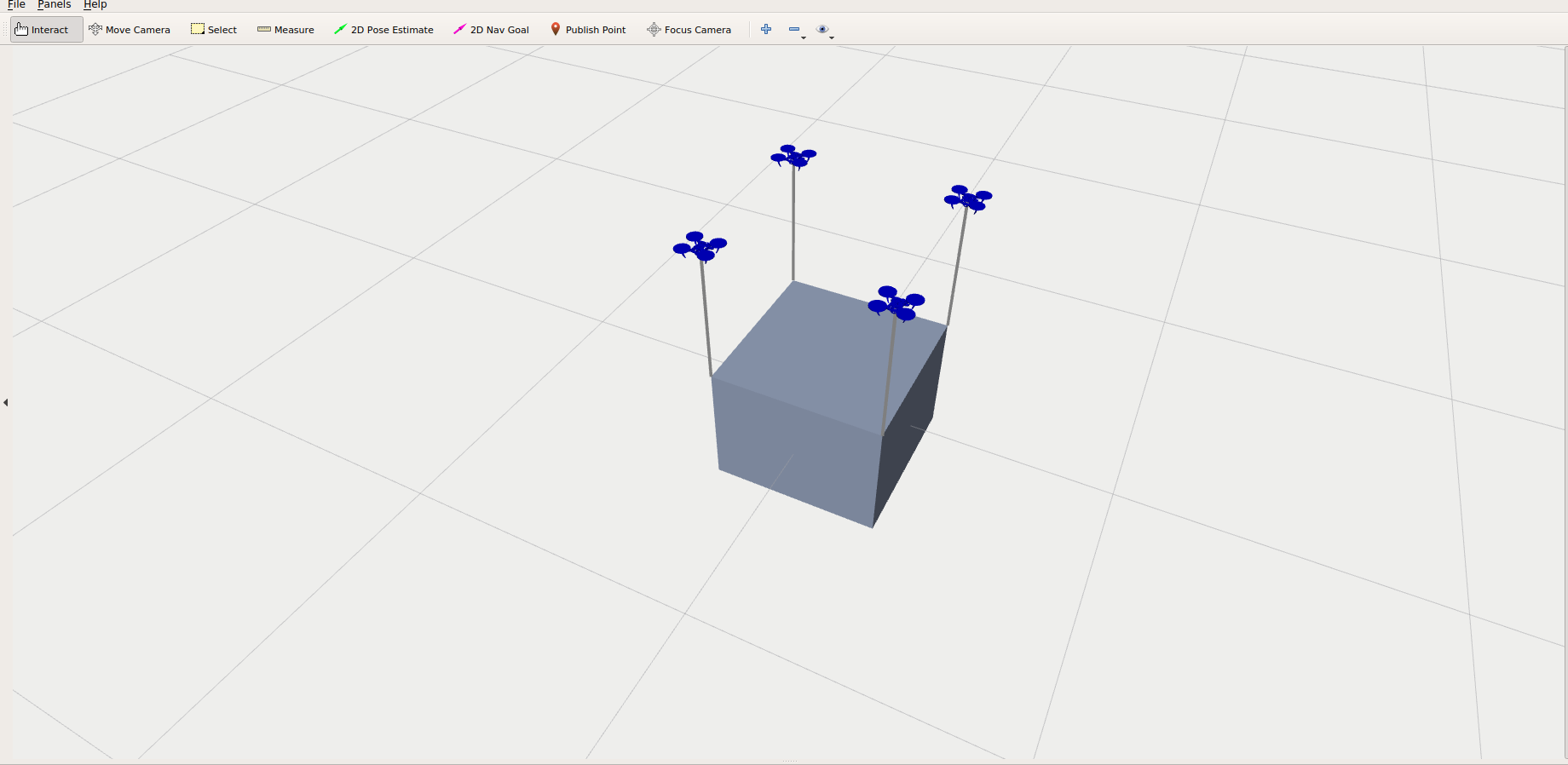}
  \caption{Visualization of the Cooperative Transportation and
Manipulation of Cable Suspended Payloads with Multiple Quadrotors in ROS and Gazebo}  \label{Fig1}
\end{figure} 
\begin{figure}[h]
\captionsetup{justification=centering}
 \centering \includegraphics[width=0.4 \textwidth]{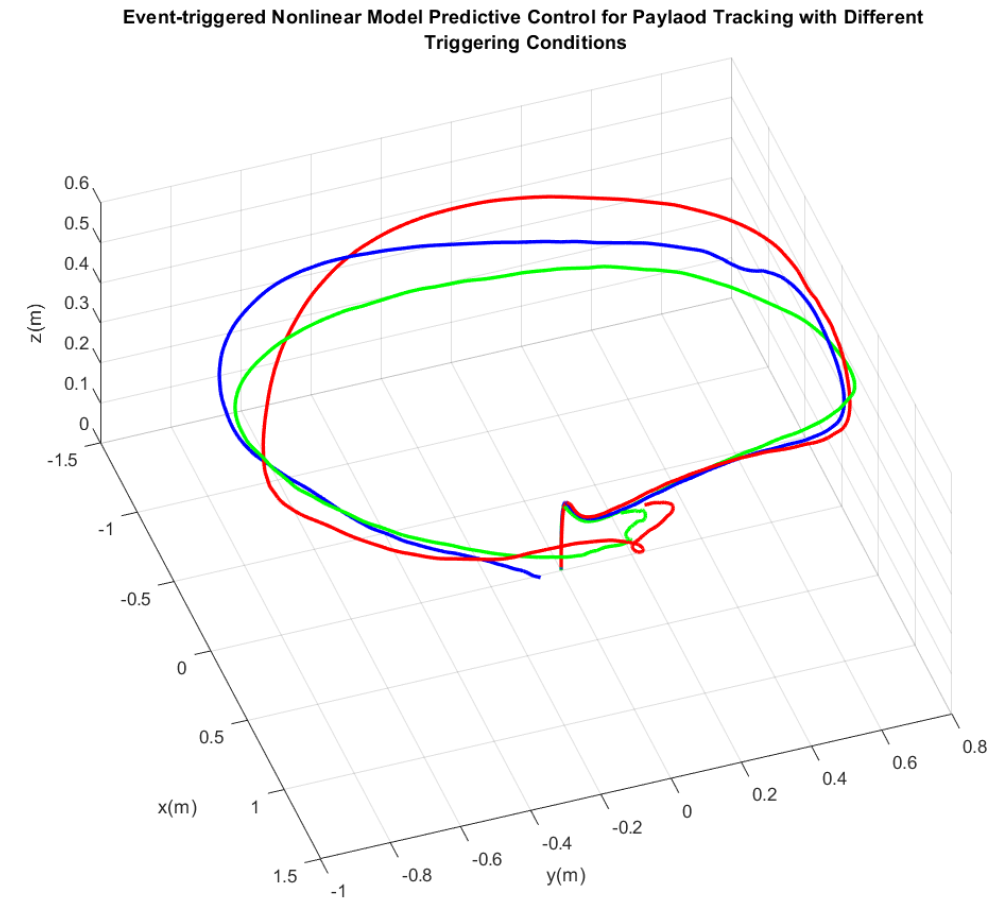}
  \caption{Event-triggered nonlinear model predictive control for position trajectory of payload in 3D with different 
triggering conditions}  \label{Fig1}
\end{figure} 
In this section, we demonstrate the effectiveness of our event-based NMPC method through simulation results. We describe our simulation setup, including the environments and platforms used, and present results for four quadrotors transporting a rigid-body payload along a circular trajectory. 
\begin{figure} [h]
    \centering
    \begin{subfigure}[b]{0.23\textwidth}
        \includegraphics[width=\textwidth]{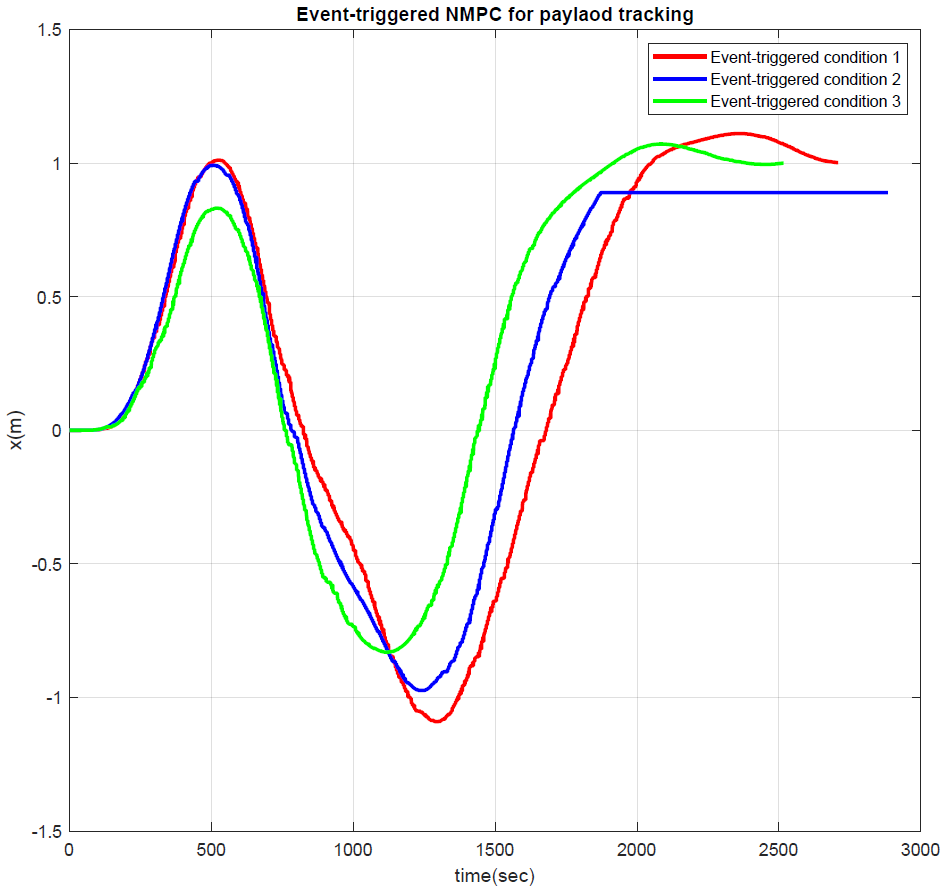}
        \caption{}
        \label{fig:gull}
    \end{subfigure}
    ~ %add desired spacing between images, e. g. ~, \quad, \qquad, \hfill etc. 
    \begin{subfigure}[b]{0.23\textwidth}
        \includegraphics[width=\textwidth]{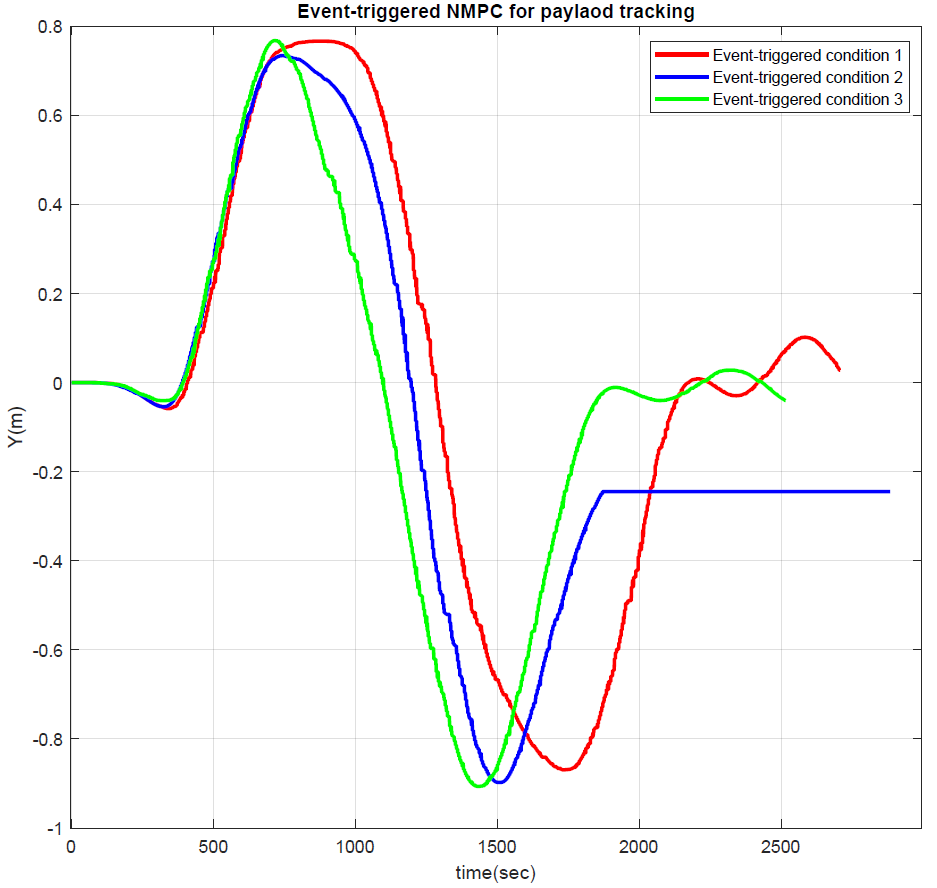}
        \caption{}
        \label{fig:tiger}
    \end{subfigure}
    ~ %add desired spacing between images, e. g. ~, \quad, \qquad, \hfill etc. 
    \begin{subfigure}[b]{0.23\textwidth}
        \includegraphics[width=\textwidth]{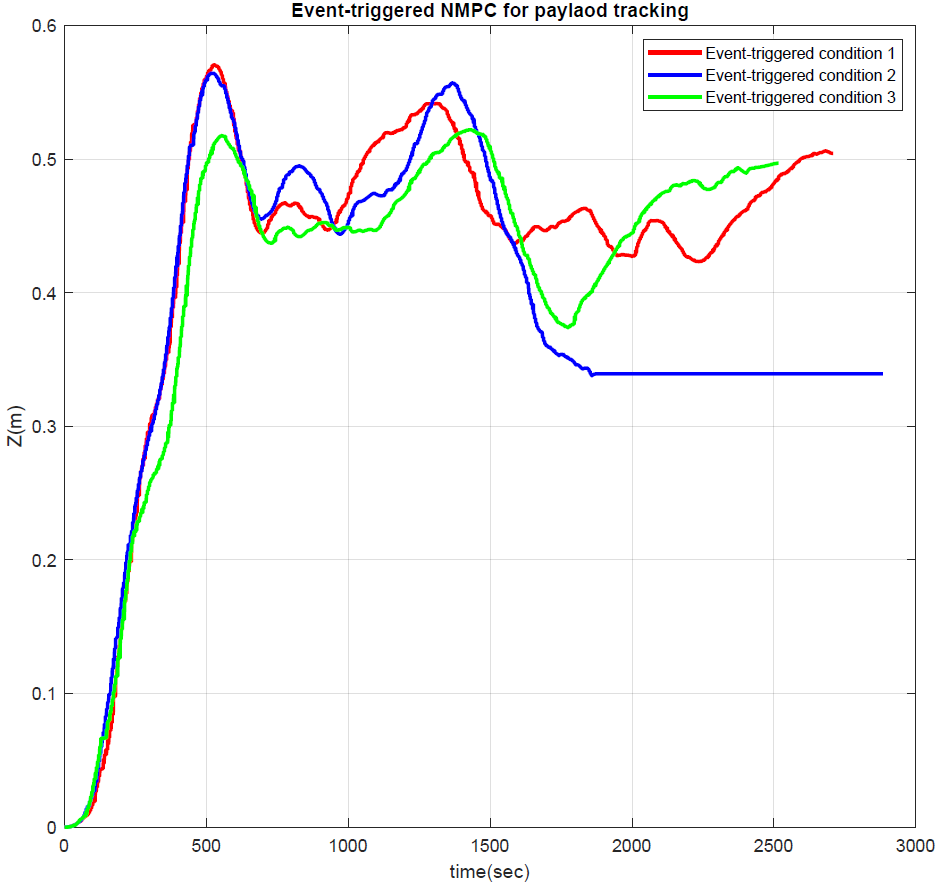}
        \caption{}
        \label{fig:mouse}
    \end{subfigure}
    \caption{Event-triggered nonlinear model predictive control for position trajectory of payload with different 
triggering conditions} \label{Fig5}
\end{figure} 
\begin{figure} [h]
    \centering
    \begin{subfigure}[b]{0.23\textwidth}
        \includegraphics[width=\textwidth]{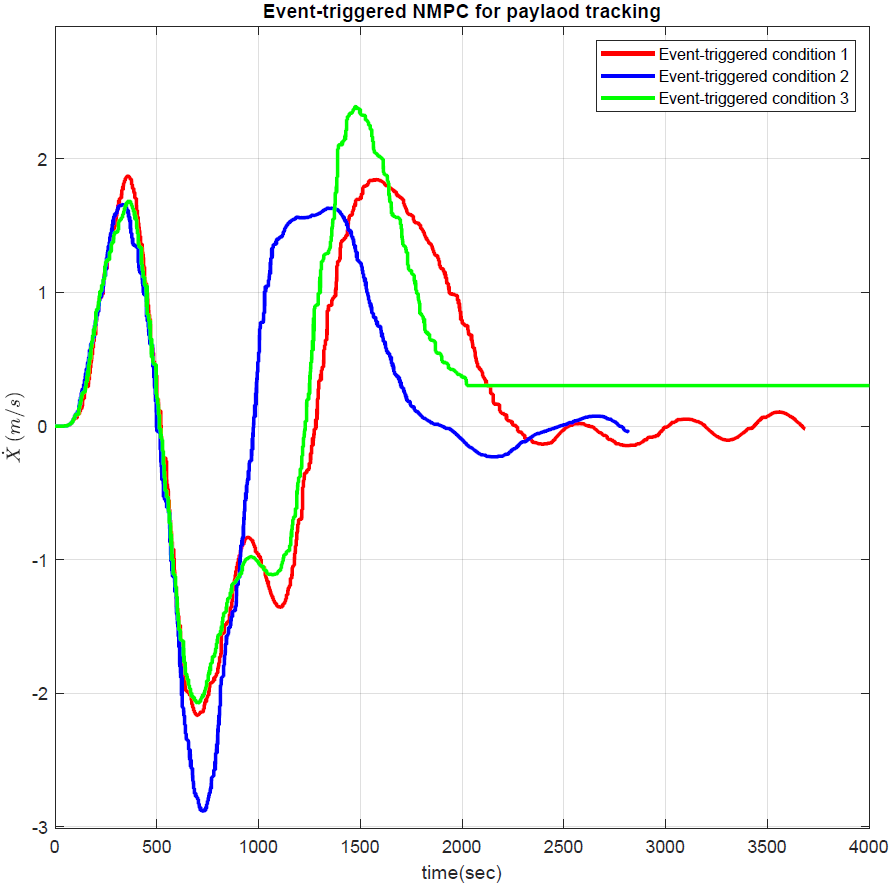}
        \caption{}
        \label{fig:gull}
    \end{subfigure}
    ~ %add desired spacing between images, e. g. ~, \quad, \qquad, \hfill etc. 
    \begin{subfigure}[b]{0.23\textwidth}
        \includegraphics[width=\textwidth]{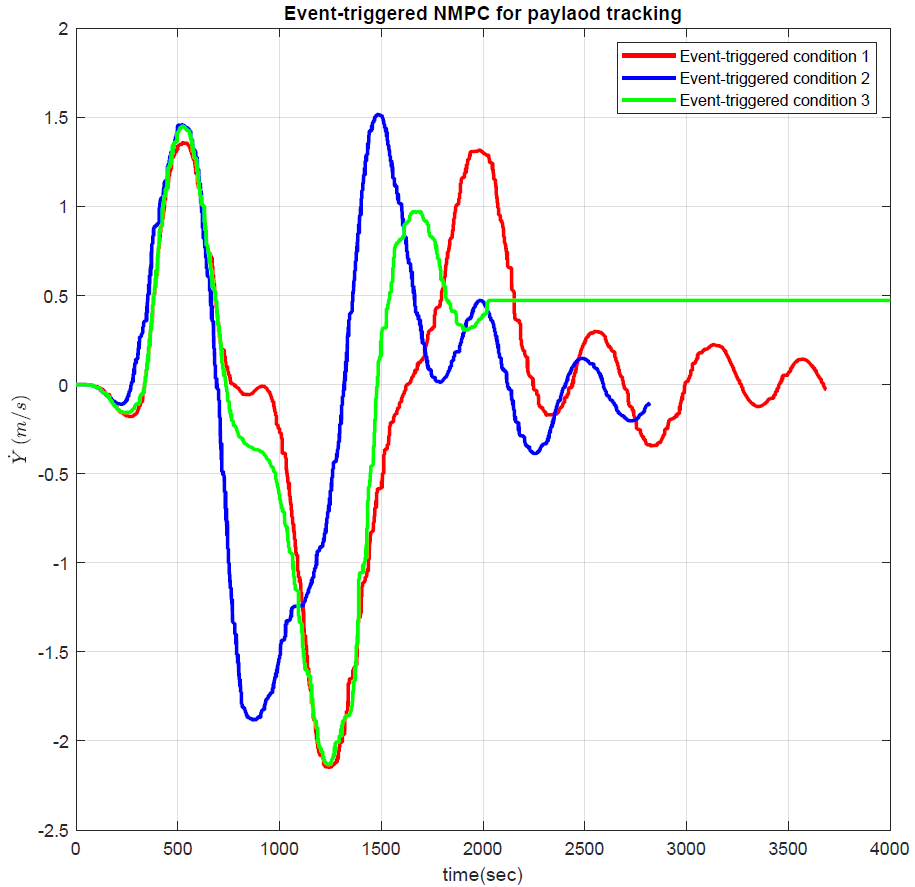}
        \caption{}
        \label{fig:tiger}
    \end{subfigure}
    ~ %add desired spacing between images, e. g. ~, \quad, \qquad, \hfill etc. 
    \begin{subfigure}[b]{0.23\textwidth}
        \includegraphics[width=\textwidth]{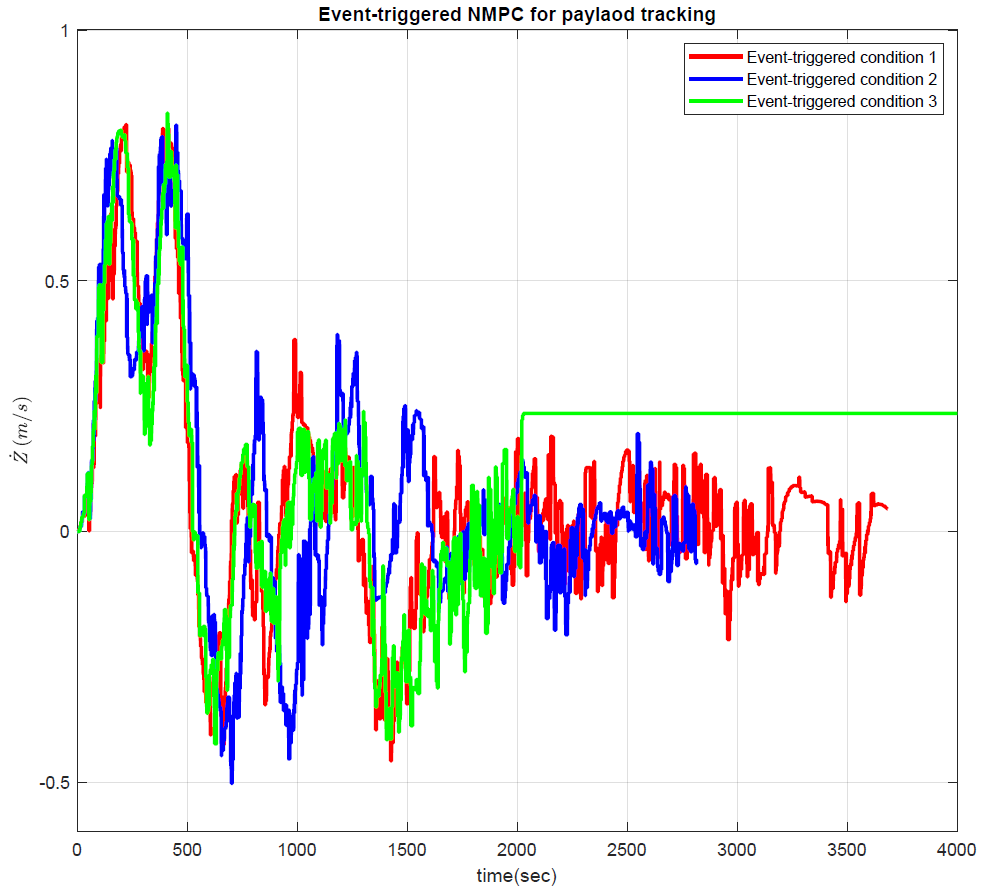}
        \caption{}
        \label{fig:mouse}
    \end{subfigure}
    \caption{Event-triggered nonlinear model predictive control for velocity trajectory of payload with different 
triggering conditions} \label{Fig5}
\end{figure} 
\begin{figure}[h]
\captionsetup{justification=centering}
 \centering \includegraphics[width=0.5 \textwidth]{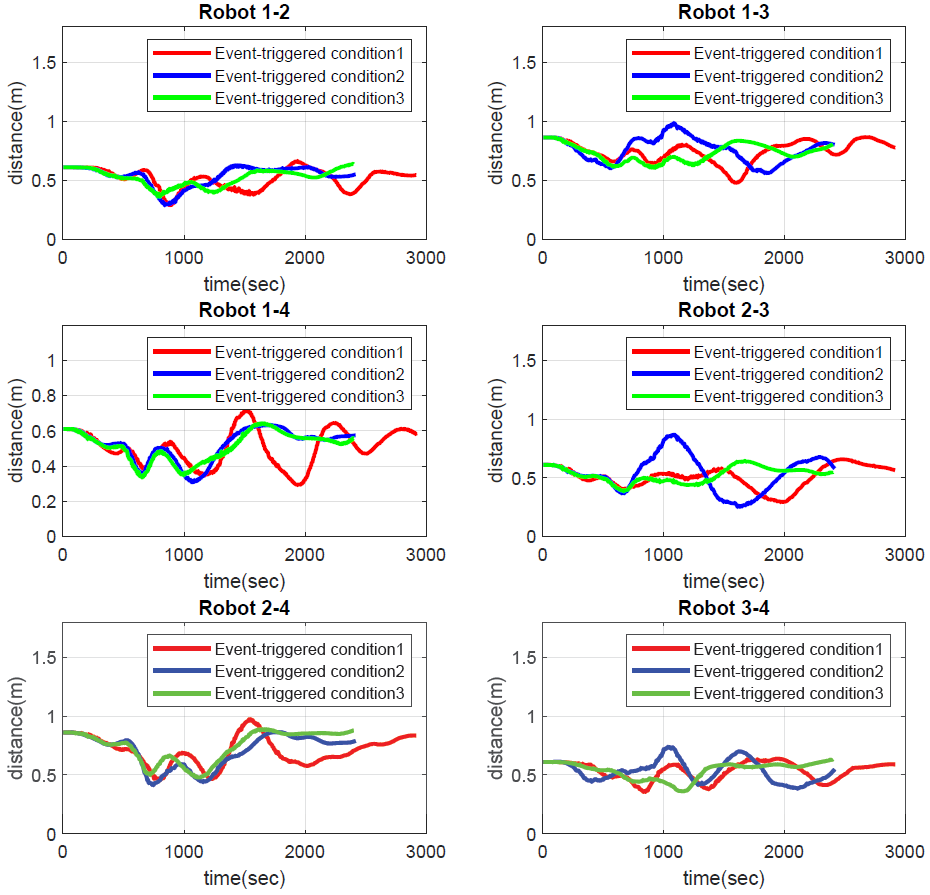}
  \caption{Inter-robot distance in the circular trajectory
tracking in simulation, staying at the boundary of 1 m}  \label{Fig1}
\end{figure} 
\subsection{Simulation Setup}
We conducted our simulations by considering three different event-triggering conditions and tuning the NMPC parameters, including the horizon and weights. The simulations were performed in a virtual environment created in Gazebo, with a rectangular payload suspended from four quadrotors, each connected by a 1-meter cable (\(l_k = 1\) m).  The payload mass was 232 g, exceeding the capacity of each individual vehicle. Fig. 3 represent the visualization of the cooperative
transportation and manipulation of cable suspended payloads with multiple quadrotors in ROS and Gazebo.
\subsection{Trajectory Tracking}

In this section, we evaluate the performance of our proposed event-based NMPC method in trajectory-tracking simulations. The results are shown in Figs. 4, 5 and 6. We consider a circular trajectory in our simulations. The circular trajectory is defined as

\[
x_{c,L,\text{des}}(t) = \begin{bmatrix} r \cos \left(\frac{2\pi t}{T_c}\right) \\ r \sin \left(\frac{2\pi t}{T_c}\right) \\ h \end{bmatrix}^\top,
\]

with a period \(T_c = 15\) s, radius \(r = 1.0\) m, and a constant height \(h = 0.5\) m. The rectangular trajectory has dimensions of 2 m in the \(x\) direction and 1 m in the \(y\) direction. We consider three different event-triggering conditions and evaluate the performance of NMPC control and the number of executions of the optimization algorithm. Note that the number of executions of the optimization algorithm (NMPC) for event-triggered conditions 1, 2, and 3 are 25, 33, and 49, respectively.
  In Figures 4,  5, and 6, we can observe that the position and velocity of the center of mass of the payload in the green graph, which has a higher number of triggerings, results in better tracking of the desired trajectory and faster convergence. However, this comes at the cost of a higher number of executions of the optimization algorithm. Conversely, the position of the center of mass of the payload in the red graph, which has fewer triggerings, leads to slower convergence in trajectory tracking. The blue graph, on the other hand, demonstrates a balance between control performance and the number of triggerings, offering a compromise between the extremes seen in the green and red graphs.  Before the NMPC begins operation, the distances between robots 1-2, robots 1-4, robots 2-3, and robots 3-4 are approximately 0.6 meters, and the distances between robots 1-3 and robots 2-4 are approximately 0.86 meters, all of which are below the 1-meter threshold. Once the NMPC is activated, Figure 7 shows the inter-robot distances based on the different event-triggered conditions. It can be seen that the inter-robot distances based on condition 3 have fewer oscillations because the number of executions of the NMPC optimization is higher than the others. This is because the optimization process aims to minimize the square norm of the null space coefficients to conserve energy. Additionally, all inter-robot distances eventually converge to the initial distances. However, these simulation results show that we can achieve a balance between scalability and performance by tuning event-triggered parameters. Additionally, the tension constraints implicitly ensure that the actuator constraints are met.

\section{CONCLUSION}
In this paper, we have presented a novel event-triggered distributed Nonlinear Model Predictive Control (NMPC) method for cooperative transportation of cable-suspended payloads using multiple quadrotors. Our approach addresses critical challenges such as payload manipulation, inter-robot separation, obstacle avoidance, and trajectory tracking, all while enhancing computational efficiency and energy utilization. By integrating an event-triggered mechanism, we significantly reduce unnecessary computations and communication, extending the operational range and endurance of MAVs.

The lightweight state vector parametrization and SE(3) manifold trajectory planning employed in our method ensure real-time computational feasibility, making it suitable for dynamic and resource-constrained environments. Our approach has been rigorously validated through simulations, demonstrating its robustness and scalability.

Overall, our event-triggered NMPC framework offers a robust, efficient, and scalable solution for cooperative transportation tasks in challenging environments such as warehouses and GPS-denied areas. Future work will focus on further enhancing the system's adaptability to varying payload characteristics and environmental conditions, as well as exploring additional applications in complex real-world scenarios. 

Future work will explore enhancing the proposed formulation by incorporating perception objectives or constraints to maximize the effectiveness of the method in autonomous tasks. Moreover, there remain several unresolved research queries, including the robustness of control algorithms to uncertainties like external disturbances and measurement noise. A thorough understanding of these facets can ensure the system's reliable operation in real-world settings.


\begin{thebibliography}{21}

\bibitem{1}
T. Tomic et al., “Toward a fully autonomous UAV: Research platform for indoor and outdoor urban search and rescue,” \textit{IEEE Robot. Automat. Mag.}, vol. 19, no. 3, pp. 46–56, Sep. 2012.

\bibitem{2}
T. Ozaslan et al., “Autonomous navigation and mapping for inspection of penstocks and tunnels with MAVs,” \textit{IEEE Robot. Automat. Lett.}, vol. 2, no. 3, pp. 1740–1747, Jul. 2017.

\bibitem{3}
G. Loianno, J. Thomas, and V. Kumar, “Cooperative localization and mapping of MAVs using RGB-D sensors,” in \textit{Proc. IEEE Int. Conf. Robot. Automat.}, pp. 4021–4028, May 2015.

\bibitem{4}
F. Forte, R. Naldi, and L. Marconi, “Impedance control of an aerial manipulator,” in \textit{Proc. Amer. Control Conf.}, pp. 3839–3844, 2012.

\bibitem{5}
N. Michael et al., “Collaborative mapping of an earthquake-damaged building via ground and aerial robots,” \textit{J. Field Robot.}, vol. 29, no. 5, pp. 832–841, 2012.

\bibitem{6}
G. Loianno et al., “Localization, grasping, and transportation of magnetic objects by a team of MAVs in challenging desert-like environments,” \textit{IEEE Robot. Automat. Lett.}, vol. 3, no. 3, pp. 1576–1583, Jul. 2018.

\bibitem{7}
K. Sreenath, T. Lee, and V. Kumar, “Geometric control and differential flatness of a quadrotor UAV with a cable-suspended load,” in \textit{Proc. 52nd IEEE Conf. Decis. Control.}, pp. 2269–2274, 2013.

\bibitem{8}
G. Loianno and V. Kumar, “Cooperative transportation using small quadrotors with monocular vision and inertial sensing,” \textit{IEEE Robot. Automat. Lett.}, vol. 3, no. 2, pp. 680-687, Apr. 2018.

\bibitem{9}
G. Li, R. Ge, and G. Loianno, “Cooperative transportation of cable-suspended payloads with MAVs using monocular vision and inertial sensing,” \textit{IEEE Robot. Automat. Lett.}, vol. 6, no. 3, pp. 5316-5323, Jul. 2021.

\bibitem{10}
G. Li and G. Loianno, “Nonlinear model predictive control for cooperative transportation and manipulation of cable suspended payloads with multiple quadrotors,” in \textit{2023 IEEE/RSJ Int. Conf. Intell. Robots Syst. (IROS)}, Detroit, MI, USA, pp. 5034-5041, 2023.

\bibitem{11}
X. Jin and Z. Hu, “Adaptive cooperative load transportation by a team of quadrotors with multiple constraint requirements,” \textit{IEEE Trans. Intell. Transp. Syst.}, vol. 24, no. 1, pp. 801-814, Jan. 2023.

\bibitem{12}
G. Li, X. Liu, and G. Loianno, “RotorTM: A flexible simulator for aerial transportation and manipulation,” \textit{IEEE Trans. Robot.}, vol. 40, pp. 831-850, 2024.

\bibitem{13}
T. Kargar Tasooji and H. J. Marquez, "Cooperative Localization in Mobile Robots Using Event-Triggered Mechanism: Theory and Experiments," in IEEE Transactions on Automation Science and Engineering, vol. 19, no. 4, pp. 3246-3258, Oct. 2022, doi: 10.1109/TASE.2021.3115770.

\bibitem{14}
I. K. Erunsal, J. Zheng, R. Ventura, and A. Martinoli, “Linear and nonlinear model predictive control strategies for trajectory tracking micro aerial vehicles: A comparative study,” in \textit{2022 IEEE/RSJ Int. Conf. Intell. Robots Syst. (IROS)}, Kyoto, Japan, pp. 12106-12113, 2022.

\bibitem{15}
T. Lee, “Geometric control of multiple quadrotor UAVs transporting a cable-suspended rigid body,” in \textit{53rd IEEE Conf. Decis. Control (CDC)}, pp. 6155-6160, 2014.

\bibitem{16}
Z. Sun, Y. Xia, L. Dai, and P. Campoy, “Tracking of unicycle robots using event-based MPC with adaptive prediction horizon,” \textit{IEEE-ASME Trans. Mechatronics}, vol. 25, no. 2, pp. 739-749, Apr. 2020.

\bibitem{17}
G. Li, X. Liu, and G. Loianno, “RotorTM: A flexible simulator for aerial transportation and manipulation,” 2022. [Online]. Available: https://arxiv.org/abs/2205.05140

\bibitem{18}
T. K. Tasooji and H. J. Marquez, "Event-Triggered Consensus Control for Multirobot Systems With Cooperative Localization," in IEEE Transactions on Industrial Electronics, vol. 70, no. 6, pp. 5982-5993, June 2023, doi: 10.1109/TIE.2022.3192673.

\bibitem{19}
T. K. Tasooji, S. Khodadadi and H. J. Marquez, "Event-Based Secure Consensus Control for Multirobot Systems With Cooperative Localization Against DoS Attacks," in IEEE/ASME Transactions on Mechatronics, vol. 29, no. 1, pp. 715-729, Feb. 2024, doi: 10.1109/TMECH.2023.3270819.

\bibitem{20}
T. K. Tasooji and H. J. Marquez, "Decentralized Event-Triggered Cooperative Localization in Multirobot Systems Under Random Delays: With/Without Timestamps Mechanism," in IEEE/ASME Transactions on Mechatronics, vol. 28, no. 1, pp. 555-567, Feb. 2023, doi: 10.1109/TMECH.2022.3203439.

\bibitem{21}
T. Kargar Tasooji and H. J. Marquez, "A Secure Decentralized Event-Triggered Cooperative Localization in Multi-Robot Systems Under Cyber Attack," in IEEE Access, vol. 10, pp. 128101-128121, 2022, doi: 10.1109/ACCESS.2022.3227076.

\bibitem{22}
S. Khodadadi, T. K. Tasooji and H. J. Marquez, "Observer-Based Secure Control for Vehicular Platooning Under DoS Attacks," in IEEE Access, vol. 11, pp. 20542-20552, 2023, doi: 10.1109/ACCESS.2023.3250398.

\bibitem{23}
M. A. Gozukucuk et al., "Design and Simulation of an Optimal Energy Management Strategy for Plug-In Electric Vehicles," 2018 6th International Conference on Control Engineering \& Information Technology (CEIT), Istanbul, Turkey, 2018, pp. 1-6, doi: 10.1109/CEIT.2018.8751923.

\bibitem{24}
A. Mostafazadeh, T. K. Tasooji, M. Sahin and O. Usta, "Voltage control of PV-FC-battery-wind turbine for stand-alone hybrid system based on fuzzy logic controller," 2017 10th International Conference on Electrical and Electronics Engineering (ELECO), Bursa, Turkey, 2017, pp. 170-174.

\bibitem{25}
T. K. Tasooji, A. Mostafazadeh and O. Usta, "Model predictive controller as a robust algorithm for maximum power point tracking," 2017 10th International Conference on Electrical and Electronics Engineering (ELECO), Bursa, Turkey, 2017, pp. 175-179.

\bibitem{26}
T. K. Tasooji, O. Bebek, B. Ugurlu, ”A Robust Torque Controller for
Series Elastic Actuators: Model Predictive Control with a Disturbance
Observer” Turkish National Conference on Automatic Control (TOK),
Istanbul, Turkey pp. 398-402, 2017

\bibitem{27}
T. K. Tasooji, ”Energy consumption modeling and optimization of speed profile for plug-in electric vehicles”, M.Sc. dissertation, Ozyegin Univ, Istanbul, Turkey, 2018  

\bibitem{28}
T. K. Tasooji, ”Cooperative Localization and Control In Multi-Robot Systems With Event-Triggered Mechanism: Theory and Experiments”, Ph.D. dissertation, Univ. Alberta, Edmonton, AB, Canada, 2023

\bibitem{29}
A. Sagale, T. K. Tasooji, and R. Parasuraman, “DCL-sparse: Distributed range-only cooperative localization of multi-robots in noisy and sparse sensing graphs,” arXiv [cs.RO], 2024.

\bibitem{30}
S. Khodadadi, "Observer-Based Secure Control of Vehicular Platooning Under DoS attacks", M.Sc. dissertation, Univ. Alberta, Edmonton, AB, Canada, 2022


\end{thebibliography}
\end{document}